\documentclass[journal]{IEEEtran}
\usepackage{graphicx}
\usepackage{epstopdf}
\usepackage{amsfonts}
\usepackage{amsmath}
\usepackage{amssymb}
\usepackage{booktabs,multirow,diagbox}
\usepackage{cite}
\usepackage{xcolor}
\usepackage[justification=centering]{caption}
\usepackage{subfigure}

\newtheorem{theorem}{Theorem}[section]

\newtheorem{proof}{Proof}[section]
\ifCLASSINFOpdf
\else
\fi

\begin{document}
\bibliographystyle{IEEEtran}
\title{Cooperative Jamming and Relay Selection for Covert Communications}

%
%
%

\author{Chan Gao, Bin Yang, Dong Zheng, Xiaohong Jiang, Tarik Taleb
\thanks{C. Gao is with the National Engineering Laboratory for Wireless Security, School of Cybersecurity, Xi'an University of Posts and Telecommunications, Xi'an 710121, China. E-mail: gaochan001@163.com. }
\thanks{B. Yang is with the School of Computer and Information Engineering, Chuzhou University, China. E-mail: yangbinchi@gmail.com.}
\thanks{D. Zheng is with the National Engineering Laboratory for Wireless Security, School of Cybersecurity, Xi'an University of Posts and Telecommunications, Xi'an 710121, China. E-mail:zhengdong@xupt.edu.cn.}
\thanks{X. Jiang is with the School of Systems Information Science, Future University Hakodate, Japan. E-mail: jiang@fun.ac.jp.}
\thanks{T. Taleb is with the Information Technology and Electrical Engineering, University of Oulu, Oulu 90570, Finland, and also with the Department of Computer and Information Security, Sejong University, Seoul 05006, South Korea. E-mail: tarik.taleb@oulu.fi.}
\thanks{An earlier version of part work in this paper was appeared in the author's thesis~\cite{gao}.} }

\markboth{IEEE Transactions on Communications}%
{Submitted paper}
%



\maketitle

\begin{abstract}
This paper investigates the covert communications via cooperative jamming and relay selection in a wireless relay system, where a source intends to transmit a message to its destination with the help of a selected relay, and a warden attempts to detect the existence of wireless transmissions from both the source and relay, while friendly jammers send jamming signals to prevent warden from detecting the transmission process. To this end, we first propose two relay selection schemes, namely random relay selection (RRS) and max-min relay selection (MMRS), as well as their corresponding cooperative jamming (CJ) schemes for ensuring covertness in the system. We then provide theoretical modeling for the covert rate performance under each relay selection scheme and its CJ scheme and further explore the optimal transmit power controls of both the source and relay for covert rate maximization.
Finally, extensive simulation/numerical results are presented to validate our theoretical models and also to illustrate the covert rate performance of the relay system under
cooperative jamming and relay selection.
\end{abstract}


%
\IEEEpeerreviewmaketitle

\section{Introduction}
Wireless communication technologies have fundamentally transformed our daily life in the past decade, and are expected to create a
fully-connected digital world in the coming sixth-generation (6G) era, where enabling the Internet of everything
will promote unprecedented transmissions of sensitive personal data over wireless channels\cite{6G}.
Due to the broadcasting and open characteristics of wireless channels, wireless systems are highly vulnerable to security threats both in civil and military applications.
To cope with such threats, it is desired to explore a promising security method providing strong protection for
numerous security-sensitive applications in 5G/6G wireless systems.

The available security methods mainly utilize encryption technologies implemented at upper-layer protocol.
Such methods usually require high computational power because of their complexity\cite{yan2}.
However, there exist many Internet-of-Things (IoT) devices with limited power.
As a complementary to the encryption technologies, physical layer security (PLS) is emerging as a promising class of technologies,
which are to exploit the interference and noise of wireless channels to ensure the secrecy of communications.
Specially, covert communications are a promising PLS technology
aiming to hide the process of wireless transmission from being detected by a warden.

The available works on the studies of covert communications mainly focus on one-hop and two-hop wireless relay systems \cite{yan3,10,11,12,14,15,16,17,19,20,21,25,Soltani,22,23,24,Liang,27,Aided,Ma,yan1,yan4,29,30,31,32,Forouzesh} (see Related Works in Section II),
where a transmitter attempts to covertly transmit information to a receiver with/without the help of a relay.
For the one-hop wireless systems, these works explore the fundamental covert performances in terms of covert rate and detection error probability under various scenarios
such as different channel models \cite{14,15,16,17,yan3}, channel uncertainty \cite{19}, noise uncertainty \cite{20}, channel inversion power control \cite{21}, delay constraints \cite{25}, jamming signals \cite{Soltani,22,23,24,Liang,27,Aided,Ma}, and unmanned aerial vehicle (UAV) scenarios \cite{yan1}.
The covert performances are further investigated in the two-top wireless relay systems \cite{29,30,31,32,Forouzesh}, where one/multiple wardens try to detect the presence of wireless communications from a transmitter to a relay and from the relay to a destination.

\color{blue}


\color{black}
It is notable that cooperative jamming and relay selection are two critical schemes for improving covert performances.
Using the cooperative jamming scheme, the available works mainly utilize the
jamming signals to confuse the detection of the warden, while ignoring the serious interference
of the jamming signals on the legitimate receiver which may lead to the degradation of system performances.
Hence, one fundamental issue is to design a cooperative jamming scheme such that
the jamming signals can interfere with an illegal warden and reduce the interference to the legal
receiver as much as possible. On the other hand, the relay selection is of great importance for the
improvement of covert performances.
\textcolor{black}{Specifically, the work in~\cite{Gao,9241054} demonstrates the potential of
relay selection in enhancing covert performances.  Actually, these two
schemes have major impacts on covert performances.
Recently, a joint cooperative jamming and relay selection scheme is proposed to
ensure the covert communication of the second hop transmission in a two-hop wireless relay system~\cite{Jiang},
which selects one relay and one jammer with maximal channel gain and minimal one to the destination,
respectively. Although this work can motivate us to explore the joint scheme design,
many important issues have not been investigated in such a system. For instance, both the covert communications of the two hop transmissions are generally required,
multiple idle relays can also be selected as jammers to further enhance the covert performances,
and the channel gain of the first hop transmission should be carefully considered into the relay selection since the channel gain can significantly affect the covert performances.
To address these critical issues, this paper
explores a joint design for the cooperative jamming and relay selection in wireless relay systems with careful consideration of the two hop covert communications,
multiple jammers and two hop channel gains. Furthermore, the theoretical models are developed to
characterize the covert rate performance in the systems.
The main contributions of
this paper can be summarized as follows.
}

 \begin{itemize}
\item
We consider a wireless relay system consisting of one source Alice, a number of potential relays, one destination Bob, and one warden.
In such a scenario, we propose two relay selection schemes, namely random relay selection (RRS) and max-min relay selection (MMRS),
as well as their corresponding cooperative jamming schemes for ensuring covertness.

\item
By applying a joint jamming and RRS scheme, we first examine the transmission strategy design for the source and determine the detection error probability at the warden. We then derive the expressions for three performance metrics (i.e., transmission outage probability, detection error probability of the warden, and covert rate), and also explore the covert rate maximization through efficient numerical searches under the given covertness and outage requirements.
\item
We further apply a joint jamming and MMRS scheme. Under this scheme,
we first examine the transmission strategy design for the source.
We then derive the detection error probability of the warden, and optimize the transmit power of the source to maximize the covert rate with a constraint of covertness requirement through efficient numerical searches.

\item
Finally, extensive simulation and numerical results are presented to validate our theoretical models and also to illustrate the covert rate performance of the relay system under joint jamming and relay selection.
\end{itemize}

We organize the rest of the paper as follows.
Section II introduces related works. The system model and performance metrics are presented in Section III.
Section IV explores the covert rate performance under the joint jamming and RRS scheme.
Section V explores the covert rate performance under the joint jamming and MMRS scheme.
We provide the numerical results in Section VI.
This paper is concluded in Section VII.

\section{Related Works}

Regarding the one-hop wireless systems, the works focus on covert communications with/without the aid of jamming signals.
Without the jamming signals, the authors in~\cite{10,11,12} prove that when the number of channel uses $n$ goes to infinity, $\mathcal{O}(\sqrt{n})$ bits of message can be transmitted covertly to a legitimate receiver.
Following these works, the same results are proved to be achievable under various channel models such as discrete memoryless channels \cite{14,15}, multiple-access channels \cite{16} and state-dependent channels \cite{17}.
In addition, channel uncertainty \cite{19} and noise uncertainty \cite{20} are used to enhance covert performances.
Later, the work in \cite{21} adopts channel inversion power control to achieve covert communications.
The work in \cite{25} further explores the impact of delay constraints on covert communications.
The authors in \cite{yan4} examines delay-intolerant covert
communications in additive white Gaussian noise (AWGN) channels with a finite block length and enhances the
covert rate by using uniformly distributed random transmit power. Then, they study the optimality of Gaussian signalling for covert communications with an upper bound \cite{yan3}, and further explores the jointly optimizes the flying location and wireless communication transmit power for a UAV conducting covert operations in \cite{yan1} .


For the one-hop wireless systems with the jamming signals,
the works in \cite{Soltani,22} explore that a friendly node sends artificial noise to confuse the detection of a warden.
For the scenario of multiple interferers, the work in \cite{23} investigates the impact of the density and the transmit power of the interferers
on the covert performance, where the locations of the interferers follow the Poisson point process.
The work in \cite{24} further optimizes the covert rate through the jamming signals from the interferers in the scenario consisting of
a source equipped with multiple antennas, a destination, randomly distributed wardens, and interferers.
The work in \cite{Liang} indicates that the covert communications are achievable via artificial noise from a friendly unmanned aerial vehicle.
In a device-to-device (D2D) underlaid cellular system, the covert communications are proved to be achievable with the aid of artificial noise from a base station (BS) \cite{27}.
Recently, each friendly jamming node can be selected to independently transmit jamming signals to defeat the warden based on an uncoordinated jammer selection scheme \cite{Aided}.
The authors in \cite{Ma} use the inherent uncertainty of backscatter transmissions to achieve active and passive covert communications.

Regarding the two-hop wireless relay systems, the existing works mainly focus on the scenario including only one relay without jamming signals.
The work in \cite{29} studied the covert performances in terms of the detection error probability and covert rate under the AWGN channels.
The channel uncertainty is used to degrade the detection of warden in \cite{30}.
The authors in \cite{31} investigate the achievable performance of covert communication in a greedy relay-aided wireless system, where
the relay also attempts to covertly send its own messages to a destination when it forwards the messages from a source.
In the work \cite{32}, they explore the performance of covert communication and associated costs for a self-sustained relay, where the source provides energy to the relay for forwarding its information and the relay's covert transmission is forbidden.
 Recently, the authors in \cite{Forouzesh} study the covert communication and secure transmission in the scenario with multiple untrusted relays,
 where the destination and the source can inject jamming signals for achieving covert communications.
\color{black}

\begin{figure}[t]
\centering
\includegraphics[width=0.75\textwidth]{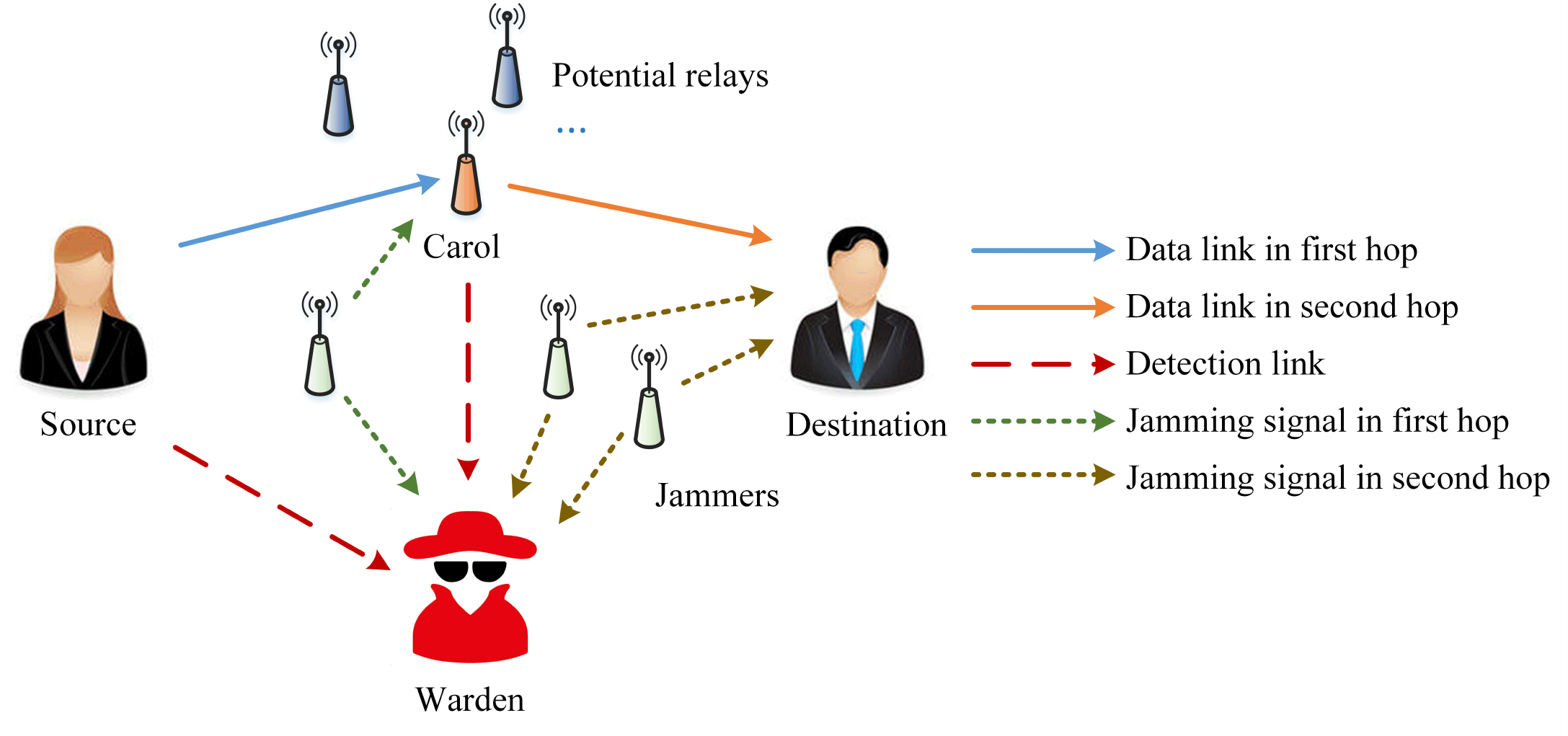}\\
\DeclareGraphicsExtensions.
\caption{Covert communication scenario.}
\label{Fig.1}
\end{figure}

\section{System Models}

\subsection{Network Model}
As shown in Fig.~\ref{Fig.1}, we consider a wireless relay system composed of
one source Alice ($A$), $n$ potential relays, one destination Bob ($B$), and one warden Willie ($W$).
Alice aims to covertly transmit a message to Bob with the aid of a relay Carol ($C$) selected
from these relays,
while Willie attempts to detect whether Alice sends a message or not.
The potential relays can also be selected as friendly jammers broadcasting jamming signals to confuse the detection of Willie.
Alice and Carol employ the same covert transmit power $P_T$ to send a message and all friendly Jammers
have a transmit power $P_J$, which is no more than a maximum power constraint $P{_{max}}$.
We assume that each of Alice, Carol, Jammers, and Bob is equipped with a single antenna.

\subsection{Channel Model}
We consider a time-slotted relay system, where a quasi-static Rayleigh fading is used to
model wireless channels. With the fading, all channel coefficients keep unchanged with one time slot and
change independently from one time slot to another.
The channel fading coefficients between Alice and Carol, Alice and Willie, Carol and Bob, Carol and Willie,
any friendly Jammer ($J_i$) and Willie, $J_i$ and Carol, and $J_i$ and Bob are denoted as
$h_{AC}$, $h_{AW}$, $h_{CB}$, $h_{CW}$, $h_{J_iW}$, $h_{J_iC}$, and $h_{J_iB}$, respectively, which
follow a complex Gaussian distribution with zero mean and unit variance.
The $|h{_{ij}}|^2$ is the corresponding channel gain, where $ij \in \{AC, AW, CB, CW, J_iC, J_iW,J_iB\}$.
We use AWGN with variance $\sigma^2$ to model the channel noise.
We assume that Carol works in half-duplex mode and hence the transmission from Alice to Carol and
that from Carol to Bob occur in different time slots.
Without loss of generality, we assume that the total system bandwidth is 1 MHz.

\subsection{Relay Selection Schemes}
We consider two relay selection schemes (i.e., RRS and MMRS) in our study.

\textbf{RRS:}
Under such a scheme, Alice randomly chooses one from all potential relays, which will help him to forward the message to Bob.

\textbf{MMRS:}
Under this scheme,
a potential relay can be selected as the relay Carol
if the following condition holds:
the maximum value of the minimum channel gain between
$|h_{AC}|^{2}$ and $|h_{CB}|^{2}$
equals the maximum one of all minimum channel gains.
Here, each minimum channel gain corresponds to the minimum one between
$|h_{Ai}|^{2}$ and $|h_{{i}B}|^{2}$ for any potential relay $i$.

\subsection{Cooperative Jamming Scheme}
Under each relay selection scheme, its corresponding cooperative jamming (CJ) scheme is further proposed for enhancing covert performance.
With such a scheme,
the potential relays (except Carol) may be selected as jammers that generate artificial noise to confuse the detection of a warden,
while reducing the interference to Carol and Bob as much as possible
according to the following jammer selection scheme.
For the first transmission phase, any potential relay $J_i$ can serve as a jammer
only if the channel gain from $J_i$ to Carol $C$ is smaller than a threshold $\alpha$, i.e.,
${|h_{{J_i}{C}}|^{2}}<\alpha$, where $J_i$ is not Carol.
As for the second transmission phase,
$J_i$ serves as a jammer if the channel gain from $J_i$ to Bob is smaller than $\alpha$,
i.e.,${|h_{{J_i}{B}}|^{2}}<\alpha$.

\subsection{Performance Metrics}\label{PerformanceMetrics}

Willie attempts to decide whether Alice sends a message or not.
To this end, it performs two hypotheses, namely, null hypothesis $H_0$
and alternative hypothesis $H_1$.
Under $H_0$, the transmitter does not transmit a message, while it transmits under $H_1$.
Then, we give the following definitions of two performance metrics.

\textbf{Detection error probability:}
It is defined as the probability $\zeta$ that Willie misjudges whether Alice sends a message or not, which equals the sum of the probability of false alarm ${\mathbb P}_{FA}$ and that of missed detection ${\mathbb P}_{MD}$.
Here, the false alarm means that Willie approves $H_1$, but $H_0$ is true actually.
The missed detection means that Willie approves $H_0$, but $H_1$ is true actually.

\textbf{Covert rate:}
It is defined as the achievable rate at which Alice can covertly send messages to Bob
while maintaining high detection error probability at Willie.


\color{black}

\section{Covert Rate under a Joint CJ and RRS Scheme}
\subsection{Detection At Willie}
At a time slot, Willie attempts to judge whether Alice transmits a message or not according to the
two hypotheses introduced in Section~\ref{PerformanceMetrics}.
Based on the hypotheses, the received signal $y_{W}$ at Willie from Alice/Carol under the joint CJ and RRS scheme is given by
\begin{flalign}
y_{W} =
&\left\{
\begin{aligned}
&\sum_{J_{i}} \sqrt{P_{J}}h_{J_{i}W}x_{j}+ n_{W}, & \text{if}\,H_{0}\,\text{is~true} \\
&\sqrt{P_{T}}h_{kW}x_{k} + \sum_{J_{i}} \sqrt{P_{J}}h_{J_{i}W}x_{j}+ n_{W},&\text{if}\,H_{1}\,\text{is~true}
\end{aligned}
\right.
\end{flalign}
where $x_j$ is the signal transmitted by jammers $J_{i}$, $x_k$ is the signal transmitted by Alice/Carol, $k\in\{A,C\}$, and $n_W$ is the AWGN at Willie with variance $\sigma{_{W}^2}$, i.e.,$n_{W} \thicksim \mathcal{CN} (0,\sigma{_{W}^2})$.

According to the Neyman-Pearson criterion,
Willie uses the following optimal decision to minimize his detection error probability~\cite{31}:
\begin{flalign}\label{V}
Y \underset{D_{0}}{\overset{D_{1}}\gtrless} \lambda,
\end{flalign}
where
$D_{0}$ and $D_{1}$ denote that Willie decides to approve $H_{0}$ and $H_{1}$, respectively,
$\lambda$ is a detection threshold, and $Y =\frac{1}{m} \sum_{i=1}^{m}{|y_{W}^{i}|}^{2}$ is the average received power at Willie in the time slot.
Here, $y_{W}^{i}$ is the received signal at Willie in $i$th channel use, and $m$ is the number of channel uses.
Considering an infinite number of channel uses in our study, we have
\begin{flalign}
Y =
&\left\{
\begin{aligned}
&\sum_{J_{i}}P_{J}{|h_{J_{i}W}|}^{2} + \sigma_{W}^{2} , &\text{if}\,H_{0}\,\text{is true} \\
&P_{T}{|h_{kW}|}^{2}+\sum_{J_{i}}P_{J}{|h_{J_{i}W}|}^{2} + \sigma_{W}^{2} . &\text{if}\,H_{1}\,\text{is true}
\end{aligned}
\right.
\end{flalign}

\subsection{Optimal Detection Threshold and Minimum Detection
Error Probability}
To determine the optimal detection threshold and minimum detection error probability,
we first derive the detection error probability at Willie given in the following Theorem.

\begin{theorem}
Under the CJ and RRS schemes, the detection error probability $\zeta$ at Willie can be determined as
\begin{flalign}\label{eq:rezeta}
&\zeta = \left\{
\begin{aligned}
&1+\frac{\Gamma(l)}{(l-1)!}-\left(\frac{P_{T}}{P_{T}-P_{J}}\right)^{l}\text{exp}\left(\frac{\sigma_{W}^{2} - \lambda}{P_{T} }\right), \text{if} \,\lambda \geq \sigma_{W}^{2}\\
&1,\,\,\,\,\,\,\quad\qquad\qquad\qquad\quad\qquad\quad\qquad\qquad\qquad\text{otherwise}
\end{aligned}
\right.
\end{flalign}
where $l$ denotes the number of friendly jammers
and gamma function $\Gamma(l)$ is given by
\begin{align}\label{eq:gammaFunction}
\Gamma(l)=\int_{0}^{\infty} \ x^{(l-1)}e^{-x}\, dx.
\end{align}
\end{theorem}
where $x\thicksim\Gamma(l,a)$, $a= (\lambda-\sigma_{W}^{2})/P_{J}$.

\begin{proof}
Based on the definition of detection error probability, we have
\begin{align}\label{eq1:zeta}
 \zeta&= \mathbb{P}_{FA} + \mathbb{P}_{MD}.
\end{align}

We first determine $\mathbb{P}_{FA}$.
We use $l$ to denote the number of friendly jammers, and then we have
\begin{align}\label{eq:fa1}
\mathbb{P}_{FA} &=P(\sum_{i=0}^{l}P_{J}{|h_{J_{i}W}|}^{2} + \sigma_{W}^{2}\geq\lambda) \nonumber\\
&= P\left(\sum_{i=0}^{l}{|h_{J_{i}W}|}^{2}\geq\frac{\lambda-\sigma_{W}^{2}}{P_{J}}\right)\nonumber\\
&=\int_{\left(\frac{\lambda-\sigma_{W}^{2}}{P_{J}}\right)}^{\infty}f_{\sum_{i=0}^{l}{|h_{J_{i}W}|}^{2}}(x)dx.
\end{align}

Since the probability density function (PDF) of the random variable ${|h_{J_{i}W}|}^{2}$
is given by
\begin{align}
f_{{|h_{J_{i}W}|}^{2}}(x)=e^{-x}, \text{if}\,0< x <\infty
\end{align}
using the convolution theorem, the PDF of $\sum_{i=0}^{l}{|h_{J_{i}W}|}^{2}$ can be determined as
\begin{align}\label{eq:midu}
f_{\sum_{i=0}^{l}{|h_{J_{i}W}|}^{2}}(x)=\frac{1}{(l-1)!}x^{(l-1)}e^{-x},\text{if}\,0< x <\infty
\end{align}

Thus, we obtain
\begin{flalign}\label{eq:FA1}
&\mathbb{P}_{FA} = \left\{
\begin{aligned}
&\frac{\Gamma(l)}{(l-1)!}, \qquad\quad\text{if} \,\lambda \geq \sigma_{W}^{2}\\
&1.\qquad\qquad\qquad\text{otherwise}
\end{aligned}
\right.
\end{flalign}

We proceed to determine $\mathbb{P}_{MD}$. We use $Z$ to denote the event that there are $l$ potential relays serving as friendly jammers.
By applying the law of total probability, we have
\begin{flalign} \label{PMD}
\mathbb{P}_{MD} &=  P(P_{T}{|h_{kW}|}^{2} + \sum_{i=0}^{l}P_{J}{|h_{J_{i}W}|}^{2} + \sigma_{W}^{2} < \lambda)\nonumber\\
&= \sum_{l=0}^{n-1}P (P_{T}{|h_{kW}|}^{2} + \sum_{i=0}^{l}P_{J}{|h_{J_{i}W}|}^{2} + \sigma_{W}^{2} < \lambda|Z)P(Z) \nonumber\\
&=\mathbb{E}_{{|h_{J_{i}W}|}^{2}}\left[1-\text{exp}\left(\frac{\sum_{i=0}^{l}P_{J}{|h_{J_{i}W}|}^{2}+\sigma_{W}^{2}-\lambda}{P_{T}}\right)\right]\nonumber\\
&=1-\mathbb{E}_{{|h_{J_{i}W}|}^{2}}\text{exp}\left(\frac{\sum_{i=0}^{l}P_{J}{|h_{J_{i}W}|}^{2}+\sigma_{W}^{2}-\lambda}{P_{T}}\right)\nonumber\\
&=1-\text{exp}\left(\frac{\sigma_{W}^{2}-\lambda}{P_{T}}\right)\prod_{i=0}^{l}\mathbb{E}_{{|h_{J_{i}W}|}^{2}}\text{exp}\left(\frac{P_{J}{|h_{J_{i}W}|}^{2}}{P_{T}}\right)\nonumber\\
&=1-\text{exp}\left(\frac{\sigma_{W}^{2}-\lambda}{P_{T}}\right)\prod_{i=0}^{l}\int_{0}^{\infty}\text{exp}\left(\frac{P_{J}{|h_{J_{i}W}|}^{2}}{P_{T}}\right)f_{{|h_{J_{i}W}|}^{2}}(x)dx,
\end{flalign}
where $\mathbb{E}(\cdot)$ denotes the expectation function and the conditional expectation function is used to derive $\mathbb{P}_{MD}$.

Substituting $f_{{|h_{J_{i}W}|}^{2}}(x)$ into~(\ref{PMD}), we obtain
\begin{flalign}\label{eq:MD1}
&\mathbb{P}_{MD} =\left\{
\begin{aligned}
&1 - \left(\frac{P_{T}}{P_{T}-P_{J}}\right)^{l}\text{exp}\left(\frac{\sigma_{W}^{2} - \lambda}{P_{T} }\right), \text{if} \,\lambda > \sigma_{W}^{2}\\
&0.\,\,\,\,\quad\quad\quad\qquad\qquad\qquad\qquad\qquad\text{otherwise}
\end{aligned}
\right.
\end{flalign}

Substituting~(\ref{eq:FA1}) and~(\ref{eq:MD1}) into~(\ref{eq1:zeta}),
(\ref{eq:rezeta}) follows.

\end{proof}


When $\lambda \leq\sigma_{W}^{2}$, $\zeta =1$. This means that Willie cannot detect the transmission from Alice to Carol and the one from Carol to Bob.

Thus, we only consider the case of $\lambda > \sigma_{W}^{2}$.
Taking the derivation of (\ref{eq:rezeta}) with respect to $\lambda$, we have
\begin{flalign}\label{eq:parzata}
\frac{\partial\zeta}{\partial\lambda} &= -\frac{(\lambda-\sigma_{W}^{2})^{(l-1)}\text{exp}(\frac{\sigma_{W}^{2} - \lambda}{P_{J} })}{P_{J}^{2}(l-1)!}\nonumber \\
&+\frac{1}{P_{T}}\left(\frac{P_{T}}{P_{T}-P_{J}}\right)^{l}\text{exp}\left(\frac{\sigma_{W}^{2} - \lambda}{P_{T} }\right).
\end{flalign}

Then, the optimal threshold $\lambda^{*}$ is the solution of $\frac{\partial\zeta}{\partial\lambda} = 0$.
Substituting $\lambda^{*}$ into (\ref{eq:rezeta}),
we obtain the minimum detection error probability $\zeta^{*}=\zeta(\lambda^{*})$.

\color{black}
\subsection{Covert Rate Modeling}
To model the fundamental covert rate performance, we first determine the transmission outage probability from Alice to Bob.
The transmission outage means that the received signal strength at the receiver Carol/Bob is smaller than its required threshold
$\theta$ so that the receiver cannot successfully recover the original message.

We derive the transmission outage probability in the following Theorem.
\begin{theorem}
We use $P_{to}$ to denote the transmission outage probability. Then, we have
\begin{flalign} \label{eq:pto}
P_{to}&=1-\text{exp}\left(-\frac{\theta(\sigma_{C}^{2}+\sigma_{B}^{2})}{P_{T}}\right)\left[\frac{1-e^{(-\alpha(1+K))}}{(1+K)(1-e^{-\alpha})}\right]^{2n-2},
\end{flalign}
where $K=\theta P_{J}/P_{T}$.
\end{theorem}

\begin{proof}
In our concerned two-hop wireless network, if the transmission is not an outage, each of the two transmissions from Alice to Carol and from Carol to Bob cannot be an outage.
Thus, the transmission outage probability $P_{to}$ is given by
\begin{flalign}\label{TPto}
P_{to}&=P(\text{SIR}_{AC}<\theta \bigcup \text{SIR}_{CB}<\theta)\nonumber \\
&=1-P(\text{SIR}_{AC}\geq \theta \bigcap \text{SIR}_{CB}\geq\theta),
\end{flalign}
where the signal-to-noise $\text{SIR}_{AC}$ at $C$ is expressed as
\begin{flalign}
\text{SIR}_{AC} = \frac{P_{T}{|h_{AC}|}^{2}}{\sum_{J_{i}}P_{J} {|h_{J_{i}C}|^{2}}+\sigma_{C}^2},
\end{flalign}
and
the signal-to-noise $\text{SIR}_{CB}$ at $B$ is expressed as
\begin{flalign}
\text{SIR}_{CB} = \frac{P_{T}{|h_{CB}|}^{2}}{\sum_{J_{i}}P_{J} {|h_{J_{i}B}|^{2}}+\sigma_{B}^2}.
\end{flalign}
Here, $\sigma_{C}^2$ and $\sigma_{B}^2$ represent the noise power.

Since the two events $\text{SIR}_{AC}<\theta$ and $\text{SIR}_{CB}<\theta$ are independent of each other,
we rewrite (\ref{TPto}) as
\begin{flalign}  \label{eq:1pto}
P_{to}&=1-P(\text{SIR}_{AC}\geq\theta)P(\text{SIR}_{CB}\geq\theta).
\end{flalign}

Since the PDF of $|h_{J_{i}C}|^{2}$ is given by
\begin{flalign}
f_{|h_{J_{i}C}|^{2}}(x)= \left\{
\begin{aligned}
&\frac{e^{-x}}{1-e^{-\alpha}},  \quad\text{if} \,0 \leq x \leq \alpha \\
&0,\qquad\qquad \text{if} \,x \geq \alpha
\end{aligned}
\right.
\end{flalign}
we have
\begin{flalign} \label{eq:1hop}
&P(\text{SIR}_{AC}\geq \theta)\nonumber \\
&=P\left[|h_{AC}|^{2}\geq\frac{\theta (\sum^{n-1}_{i=0,J_{i}\neq C}P_{J} {|h_{J_{i},C}|^{2}}+\sigma_{C}^2)}{P_{T}}\right]\nonumber \\
&=\mathbb{E}\left[\text{exp}\left(-\frac{ \theta (\sum^{n-1}_{i=0,J_{i}\neq C}P_{J} {|h_{J_{i}C}|^{2}}+\sigma_{C}^2)}{P_{T}}\right)\right]\nonumber \\
&=\text{exp}\left(-\frac{\theta \sigma_{C}^2}{P_{T}}\right)\prod^{n-1}_{i=0,J_{i}\neq C}\mathbb{E}\left[\text{exp}(-\frac{\theta P_{J}|h_{J_{i}C}|^{2}}{P_{T}})\right]\nonumber \\
&=\text{exp}\left(-\frac{\theta \sigma_{C}^2}{P_{T}}\right)\left[\frac{1-e^{(-\alpha(1+K))}}{(1+K)(1-e^{-\alpha})}\right]^{n-1},
\end{flalign}
where $K=\theta P_{J}/P_{T}$.

Similarly, we have
\begin{flalign} \label{eq:2hop}
P(\text{SIR}_{CB}\geq\theta)=&\text{exp}\left(-\frac{\theta \sigma_{B}^2}{P_{T}}\right)\left[\frac{1-e^{(-\alpha(1+K))}}{(1+K)(1-e^{-\alpha})}\right]^{n-1}.
\end{flalign}
Substituting (\ref{eq:1hop}) and (\ref{eq:2hop}) into  (\ref{eq:1pto}), we obtain (\ref{eq:pto}).
\end{proof}

Based on the $P_{to}$, we obtain the covert rate $R_{AB}$ from Alice to Bob as follows.
\begin{flalign} \label{eq:c}
R_{AB}=&(1-P_{to})\text{min}\{R_{AC},R_{CB}\},
\end{flalign}
where the achievable covert rate $R_{AC}$ from Alice to Carol is expressed as
 $R_{AC}=\text{log}_{2}(1+\text{SIR}_{AC})$,
and the achievable rate $R_{CB}$ from Carol to Bob is expressed as
$R_{CB}=\text{log}_{2}(1+\text{SIR}_{CB})$.

\color{black}

\subsection{Covert Rate Maximization}
Our goal is to maximize the covert rate $R_{AB}$ while maintaining a high detection error probability at Willie. It can be formulated as the following optimization problem.

\begin{subequations}\label{eq:optimalC1}
\begin{alignat}{2}
\text{Maximize}\quad &R_{AB} \\
\text{\textit{s.t.}} \quad \quad\quad\,\,\,\,\, &\zeta^{*}(P_{T}) \geq 1 - \varepsilon_{c},   \label{const:zeta}  \\
\, \quad \quad\quad\quad\,\,\,\, &P_{T}\leq P_{max}, \label{22c}\\
  \,\quad\quad\quad\quad\,\,\,\,&\varepsilon_{c}\in(0,1), \label{22d}
\end{alignat}\nonumber
\end{subequations}
where $\varepsilon_{c}$ represents the covert requirement,
(\ref{const:zeta}) represents the covert constraint, and (\ref{22c})
represents the range of the transmit power $P_{T}$.
The optimization problem can be solved using the stochastic gradient descent
algorithm.



\color{black}
\section{Covert Rate under a Joint CJ and MMRS Scheme}

\subsection{Detection At Willie}
Based on the two hypotheses introduced in the Section III-E, at a time slot, the received signal $y_{W}$ at Willie from Alice/Carol under the joint CJ and MMRS scheme is given by
\begin{flalign}\label{yw2}
y_{W} =
&\left\{
\begin{aligned}
&\sum_{J_{i}} \sqrt{P_{J}}h_{J_{i}W}x_{j}+ n_{W}, & \text{if}\,H_{0}\,\text{is~true} \\
&\sqrt{P_{T}}h_{kW}x_{k} + \sum_{J_{i}} \sqrt{P_{J}}h_{J_{i}W}x_{j}+ n_{W},&\text{if}\,H_{1}\,\text{is~true}
\end{aligned}
\right.
\end{flalign}
where $x_j$ is the signal transmitted by a jammer $i$, $x_k$ is the signal transmitted by Alice/Carol, $k\in\{A,C\}$,  $P_{T}$ is the transmit power of Alice/Carol, and $n_W$ is the AWGN at Willie with variance $\sigma{_{W}^2}$, i.e.,$n_{W} \thicksim \mathcal{CN} (0,\sigma{_{W}^2})$.

Based on (\ref{V}) and (\ref{yw2}), the average received power $Y$ at Willie can be determined as
\begin{flalign}\label{T}
Y =
&\left\{
\begin{aligned}
&\sum_{J_{i}}P_{J}{|h_{J_{i}W}|}^{2} + \sigma_{W}^{2} , &\text{if}\,H_{0}\,\text{is true} \\
& \frac{\theta(\sum_{J_{i}}P_{J}{|h_{J_{i}C}|}^{2} + \sigma_{C}^{2}){|h_{A,W}|}^{2}}{{|h_{AC}|}^{2}}\\
&+\sum_{J_{i}}P_{J}{|h_{J_{i}W}|}^{2} + \sigma_{W}^{2}. &\text{if}\,H_{1}\,\text{is true}
\end{aligned}
\right.
\end{flalign}


\subsection{Optimal Detection Threshold and Minimum Detection
Error Probability}
To determine the optimal detection threshold and minimum detection error probability, we first derive the detection error probability at Willie given in the following Theorem.
\begin{theorem}
Under the CJ and MMRS schemes, the detection error probability $\zeta$ at Willie can be determined as
\begin{flalign}\label{rezeta2}
&\zeta = \left\{
\begin{aligned}
&1+\frac{\Gamma(l)}{(l-1)!}-(\frac{1}{1-\varphi})^{l}\text{exp}(\frac{\varphi(\sigma_{W}^{2} - \lambda)}{P_{J}}), \ \text{if} \,\lambda \geq \sigma_{W}^{2}\\
&1,\,\qquad\qquad\qquad\qquad\quad\qquad\quad\qquad\qquad\qquad\text{otherwise}
\end{aligned}
\right.
\end{flalign}
where $l$ is the number of friendly jammers and $\varphi=(P_{J}|h_{AC}|^{2})/\theta(\sum_{i=0}^{l}P_{J}{|h_{J_{i}C}|}^{2} + \sigma_{C}^{2})$.
\end{theorem}
\begin{proof}
Based on the definition of detection error probability, we have
\begin{align}\label{eq1:zeta2}
 \zeta&= \mathbb{P}_{FA} + \mathbb{P}_{MD}.
\end{align}

Similar to the derivation process of $\mathbb{P}_{FA}$ under the RRS scheme,
$\mathbb{P}_{FA}$ under the MMRS scheme can be determined as
\begin{flalign}\label{eq:FA2}
&\mathbb{P}_{FA} = \left\{
\begin{aligned}
&\frac{\Gamma(l)}{(l-1)!}, \,\qquad\quad\text{if} \,\lambda \geq \sigma_{W}^{2},\\
&1.\qquad\qquad\qquad\,\text{otherwise}
\end{aligned}
\right.
\end{flalign}

We use $Z$ to denote the event that there are $l$ potential relays serving as friendly jammers.
Applying the law of total probability, $\mathbb{P}_{MD}$ is determined as
\begin{flalign}
\mathbb{P}_{MD} &=  P(P_{T}{|h_{AW}|}^{2} + \sum_{i=0}^{l}P_{J}{|h_{J_{i}W}|}^{2} + \sigma_{W}^{2} < \lambda)\nonumber\\
&= \sum_{l=0}^{n-1}P (P_{T}{|h_{AW}|}^{2} + \sum_{i=0}^{l}P_{J}{|h_{J_{i}W}|}^{2} + \sigma_{W}^{2} < \lambda|Z)P(Z) \nonumber\\
&=\sum_{l=0}^{n-1}P(\frac{\theta(\sum_{i=0}^{l}P_{J}{|h_{J_{i}C}|}^{2}+\sigma_{C}^{2})}{{|h_{AC}|}^{2}}{|h_{AW}|}^{2} \nonumber\\
&+\sum_{i=0}^{l}P_{J}{|h_{J_{i}W}|}^{2} + \sigma_{W}^{2} < \lambda)P(Z) \nonumber\\
&=\mathbb{E}_{{|h_{J_{i}W}|}^{2}}(1  \nonumber\\
&-\text{exp}(\frac{(\sum_{i=0}^{l}P_{J}{|h_{J_{i}W}|}^{2}+\sigma_{W}^{2}-\lambda){|h_{AC}|}^{2}}{\theta(\sum_{i=0}^{l}P_{J}{|h_{J_{i}C}|}^{2}+\sigma_{C}^{2})}))\nonumber\\
&=1-\text{exp}\frac{(\sigma_{C}^{2}-\lambda)\varphi}{P_{J}}\prod_{i=0}^{l}\mathbb{E}_{{|h_{J_{i}W}|}^{2}}\text{exp}(\sum_{i=0}^{l}{|h_{J_{i}W}|}^{2}\varphi)\nonumber\\
&=1-\text{exp}\frac{(\sigma_{C}^{2}-\lambda)\varphi}{P_{J}}\prod_{i=0}^{l}\int_{0}^{\infty}\text{exp}({|h_{J_{i}W}|}^{2}\varphi)f_{{|h_{J_{i}W}|}^{2}}(x)dx,
\end{flalign}
where $\mathbb{E}[\cdot]$ is the expectation operator.

\color{black}
The covert communication can be achieved if $\zeta \geq 1 - \varepsilon$ for any $\varepsilon > 0$.
Similarly, when $\lambda \leq \sigma_{W}^{2}$, $\zeta =1$. This means that Willie cannot detect the transmission from Alice to Carol and the one from Carol to Bob.
Hence, we consider the case of $\lambda > \sigma_{W}^{2}$.
Take the derivation of (\ref{rezeta2}) with respect to $\lambda$, we have
\begin{flalign}\label{parzeta2}
\frac{\partial\zeta}{\partial\lambda} = -\frac{(\lambda-\sigma_{W}^{2})^{(l-1)}\text{exp}(\frac{\sigma_{W}^{2} - \lambda}{P_{J} })}{P_{J}^{2}(l-1)!}+
\frac{\varphi}{P_{J}}\left(\frac{1}{1-\varphi}\right)^{l}\text{exp}\left[\frac{\varphi(\sigma_{W}^{2} - \lambda)}{P_{J}}\right].
\end{flalign}
Let $\lambda^{*}$ denote the solution of $\frac{\partial\zeta}{\partial\lambda} = 0$, then we have the optimal threshold $\lambda^{*}$ corresponding to the minimum value of $\zeta$, i.e., $ \zeta^{*}= \zeta(\lambda^{*})$.
\end{proof}

\color{black}

\subsection{Covert Rate Modeling}
Similarly, to model the fundamental covert rate performance, we first determine the transmission outage probability from Alice to Bob.
We derive the transmission outage probability under this relay selection scheme in the following Theorem.
\begin{theorem}
We use $P_{sto}$ to denote the transmission outage probability. Then, we have
\begin{flalign}\label{psto}
 P_{sto}&=P(\text{min}\{\text{SIR}_{AC},\text{SIR}_{CB}\}<\theta) \nonumber\\
&=Uk\text{exp}\left(\frac{-\theta\sigma_{C}^2}{P_{T}}\right)\left[\frac{1-e^{-(1+z)\alpha}}{(1-e^{-\alpha})(1+z)}\right]^l\nonumber\\
&+U(k-1)\text{exp}\left(\frac{-2k\theta\sigma_{C}^2}{P_{T}}\right)\left[\frac{1-e^{-(2kz+1)\alpha}}{(1-e^{-\alpha})(2kz+1)}\right]^l
\end{flalign}
where $z=\theta P_{J}/P_{T}$,$U=\sum_{0}^{n}\binom{n}{k}(-1)^{k}(\frac{1}{{2k-1}})$.
\end{theorem}

\begin{proof}
For each relay $R_{k}$ where $k = 1,2,...,n$, let $M_{k} = \text{min}\{{|h_{A{R_{k}}}|}^{2},{|h_{{R_{k}}B}|}^{2}\}$,
and $D_{k}$ denote the event that Alice select the relay. We then have this expression
\begin{equation*}
D_{k} \triangleq {\bigcap_{v = 1,v \neq k}^n} (M_{v} \leq M_{k}),
\end{equation*}
where $M_{k}$ is an exponential random variable with mean 1/2.

Next, if $M_{k} = |h_{A{R_{k}}}|^{2}$, then based on the previous work \cite{Gao} we have
\begin{flalign}\label{eq:slink}
P({|h_{A{C}}|}^{2} < x)
&= \int_{0}^{x} ne^{-t}(2e^{-t}-e^{-x})(1-e^{-2t})^{n-1}dt  \nonumber\\
&= (1-e^{-2x})^{n}-ne^{-x} \int_{0}^{x} e^{-t}(1-e^{-2t})^{n-1}dt  \nonumber\\
&= \sum_{0}^{n}\binom{n}{k}(-1)^{k}\frac{ke^{-x}+(k-1)e^{-2kx}}{2k-1}.
\end{flalign}

The transmission outage probability can be expressed as
\begin{flalign} \label{zhengpstom}
 P_{sto}&=P(\text{SIR}_{AC}<\theta) \nonumber\\  
&=Uk\text{exp}\left[-\frac{\theta(P_{J}\sum_{i=0}^{l}|h_{J_{i}C}|^{2}+\sigma_{C}^2)}{P_{T}}\right]\nonumber\\
&+U(k-1)\text{exp}\left[-2k\frac{\theta(P_{J}\sum_{i=0}^{l}|h_{J_{i}C}|^{2}+\sigma_{C}^2)}{P_{T}}\right]\nonumber\\
&=Uk\text{exp}\left(\frac{-\theta\sigma_{C}^2}{P_{T}}\right)\prod_{i=0}^{l}\mathbb{E}_{{|h_{J_{i}C}|}^{2}}\left[\text{exp}(-z\sum_{i=0}^{l}|h_{J_{i}C}|^{2})\right]\nonumber\\
&+U(k-1)\text{exp}\left(\frac{-2k\theta\sigma_{C}^2}{P_{T}}\right)\nonumber\\
&\prod_{i=0}^{l}\mathbb{E}_{{|h_{J_{i}C}|}^{2}}\left[\text{exp}(-2kz\sum_{i=0}^{l}|h_{J_{i}C}|^{2})\right],
\end{flalign}
(\ref{psto}) can be obtained.
\end{proof}

We obtain the covert rate $R'_{AB}$ from Alice to Bob under the MMRS scheme as follows.
\begin{flalign} \label{eq:c}
R'_{AB}=(1-P_{sto})\text{min}\{R_{AC},R_{CB}\},
\end{flalign}
where the achievable covert rate $R_{AC}$ from Alice to Carol is expressed as
 $R_{AC}=\text{log}_{2}(1+\text{SIR}_{AC})$,
and the achievable rate $R_{CB}$ from Carol to Bob is expressed as
$R_{CB}=\text{log}_{2}(1+\text{SIR}_{CB})$.

Although it is first determined that the link meets the transmit condition before sending the covert messages, then the covert transmission itself will not occur outage, but the setting of $\alpha$ affects the probability of the link satisfying the requirement in a time slot. We should ensure that within a certain period of time, more time slots are sent for covert messages. Hence, we should promise that covert transmission probability must be greater than a threshold to make Alice have more opportunities to send covert messages.

\subsection{Covert Rate Maximization}
The objective of covert rate maximization is to maximize the covert rate $R'_{AB}$ while maintaining an
arbitrary high detection error probability at Willie. It can be formulated as the following optimization problem.

\begin{subequations}\label{optimalC2}
\begin{alignat}{2}
&\text{Maximize}\quad R'_{AB} \\
&\text{\textit{s.t.}}  \quad \quad\quad\,\,\,\,\,\zeta^{*}(P_{T}) \geq 1 - \varepsilon_{c},  \label{const:zeta1}  \\
&\,\quad\quad\quad\quad\,\,\,\,\,P_{T}\leq P_{max}, \label{const:power1}\\
& \,\quad\quad\quad\quad\,\,\,\,\,\varepsilon_{c}\in(0,1),
\end{alignat}\nonumber
\end{subequations}
where $\varepsilon_{c}$ is the  covertness requirement.
(\ref{const:zeta1}) represents covert constraint, and (\ref{const:power1})
represents transmit power constraint for source and relay.
We also use the stochastic gradient descent algorithm to solve the optimization problem in (\ref{optimalC2}).


\section{Numerical Results}

This section first validates our theoretical models and then explores the impact of system parameters on covert rate performance.


\subsection{ Model Validation}
To ensure the efficiency of our theoretical covert rate models in (22) and (34),
we only need to validate the transmission outage probabilities under the RSS and MMRS schemes.
\begin{figure}[]
\centering
\subfigure[Transmission outage probability validation under RRS.]{
\label{Fig.sub.1}
\includegraphics[width=0.4\textwidth]{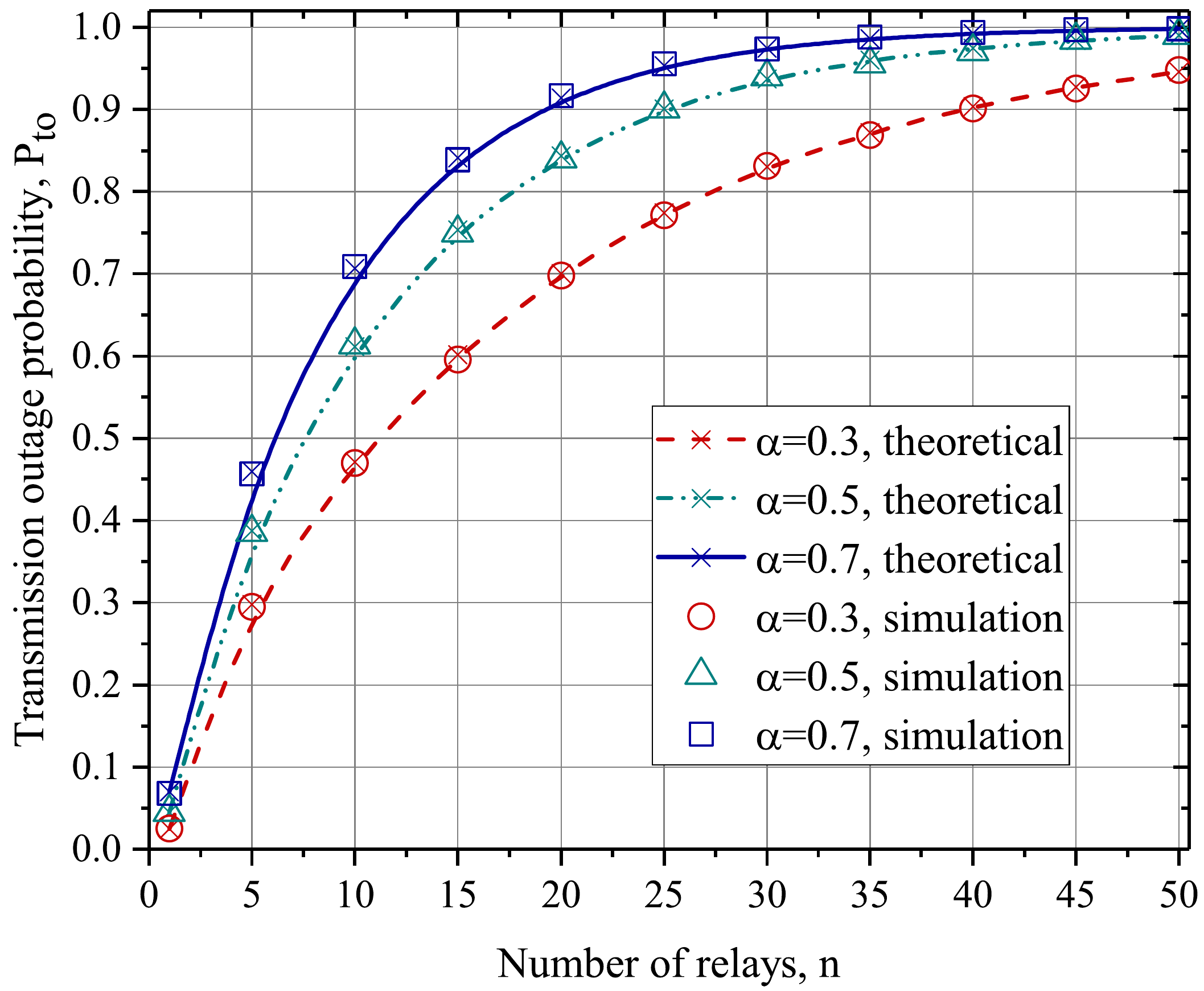}}
\subfigure[Transmission outage probability validation under MMRS.]{
\label{Fig.sub.2}
\includegraphics[width=0.4\textwidth]{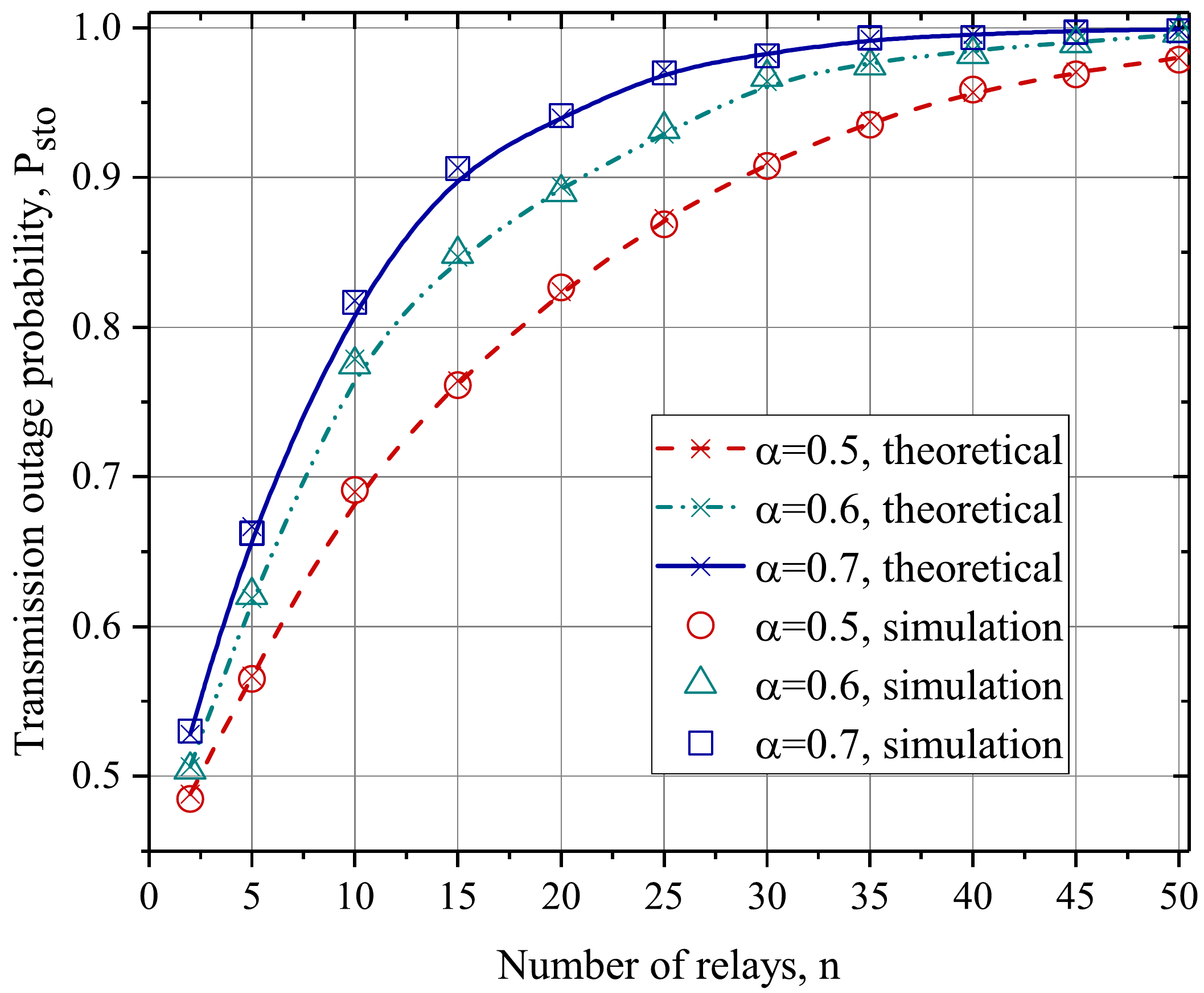}}
\caption{Transmission outage probability validation.}
\label{Fig.s}
\end{figure}
Towards this end, we compare the simulation results with the theoretical ones under these two schemes.
Specifically, the simulated transmission outage probability is calculated as the average value of
$10^5$ independent simulations. Here, the simulated probability equals the ratio of the number of transmission
outages to the total number of transmissions.
For the scenario with the setting of the number of relay $n=50$, the threshold of the cooperative jamming $\alpha = \{0.3, 0.5, 0.7\}$, transmission power of Alice and Carol $P_{T}=5$ W, outage threshold $\theta=1$,
jamming power $P_{J}=1$ W and noise $\sigma_{C}^{2}=\sigma_{B}^{2}=-5$ dB.
We can see from Fig. 2 that for each $\alpha$,
the theoretical transmission outage probability $P_{to}$ almost matches with the simulation one under
these two schemes, indicating that our theoretical model can well capture the covert rate performance under each scheme.
Another observation indicates that as the number of relays increases, $P_{to}$ increases.
Based on the cooperative jamming scheme introduced in Section III-D, we know that as the number of relays or $\alpha$ increases, the number of jamming relays increases.
This will lead to the decrease of the SIRs at the relay Carol and the destination Bob,
which further leads to the increase of $P_{to}$.

To achieve maximum covert rate based on the optimization problems of (23) and (35),
we now validate the theoretical detection error probability $\zeta$ under the two relay selection schemes
via the comparison between theoretical results and simulation ones, where
each simulated value is calculated as the average value of
$10^5$ independent simulations.
For the scenario of $n=20$, $P_{T}=5$ W, $\alpha=0.3$, $\theta=1$, $P_{J}=1$ W, $|h_{A,C}|^{2}=|h_{C,B}|^{2}$ and $\sigma_{W}^{2}=-5$ dB,
We can observe from Fig. 3 that the theoretical $\zeta$ almost matches the simulation one under each scheme. This demonstrates that our theoretical results can well predict the simulation results under these two relay selection schemes.

\begin{figure}[]
\centering
\includegraphics[width=0.4\textwidth]{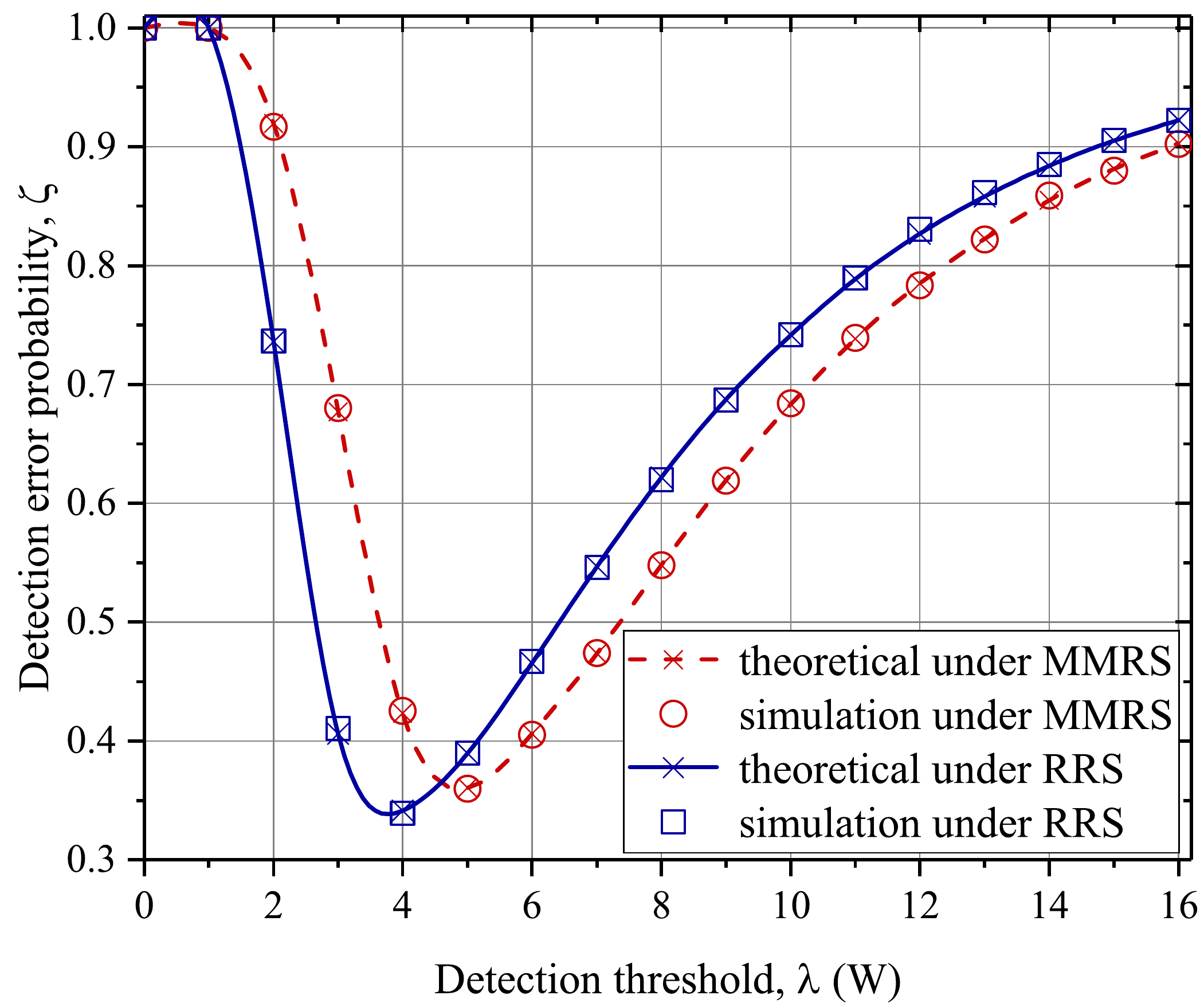}\\
\DeclareGraphicsExtensions.
\caption{Detection error probability validation.}
\label{Fig.3}
\end{figure}

We can also observe from Fig. 3 that
as $\lambda$ increases, $\zeta$ first decreases and then increases under both the schemes.
This can be explained as follows.
We know that $\zeta$ is the sum of false alarm probability $\mathbb{P}_{FA}$ and missed detection probability $\mathbb{P}_{MD}$.
Based on our theoretical analysis of these two probabilities, we know that $\mathbb{P}_{FA}$ is a decreasing function of $\lambda$ while $\mathbb{P}_{MD}$ is an increasing function.
As $\lambda$ is relatively small, the former one dominates $\zeta$, which leads to the decrease of $\zeta$ with $\lambda$.
On the other hand, as $\lambda$ further increases, the latter one dominates $\zeta$, which leads to the increase of $\zeta$.
There exists a minimum $\zeta$ corresponding to the maximum transmission power limit of  $P_{T}$, which
means that Willie has the strongest detection ability to detect the transmission of two hops.



\subsection{ Covert Performance Analysis }

We first explore the impact of $P_{T}$ on the covert rate under these two relay selection schemes.
We summarize the numerical results in Fig.4 with the setting of $\alpha = \{0.3, 0.5, 0.7\}$ and $\sigma_{C}^{2}=\sigma_{B}^{2}=-5$ dB.
It can be observed from Fig. 4 that as $P_{T}$ increases, the covert rates increase under both the schemes.
This is because the increase of $P_{T}$ leads to the increase of the SIR at the receivers Carol and Bob.
A careful observation from Fig. 4 indicates that for each fixed $P_{T}$, as $\alpha$ further increases, the covert rate will decrease.
The reason for this phenomenon is that the number of relays satisfying the selection conditions of the jammer increases, which increases the total jamming power, leading to a decrease in the SIRs at the receiver.
We can also observe that for each fixed $P_{T}$, the covert rate $R_{AB}$ under the RRS scheme in Fig. 4 (a)
is lower than that $R'_{AB}$ under the MMRS scheme in Fig.4 (b).
This is due to the following reason.
The channel quality of the two-hop transmissions under the RRS scheme is usually lower than that under the MMRS scheme,
which means that the transmission outage probability under the former is also usually higher than that under the latter.

\begin{figure}[]
\centering
\subfigure[$R_{AB}$ vs. $P_{T}$]{
\label{Fig.sub.1}
\includegraphics[width=0.4\textwidth]{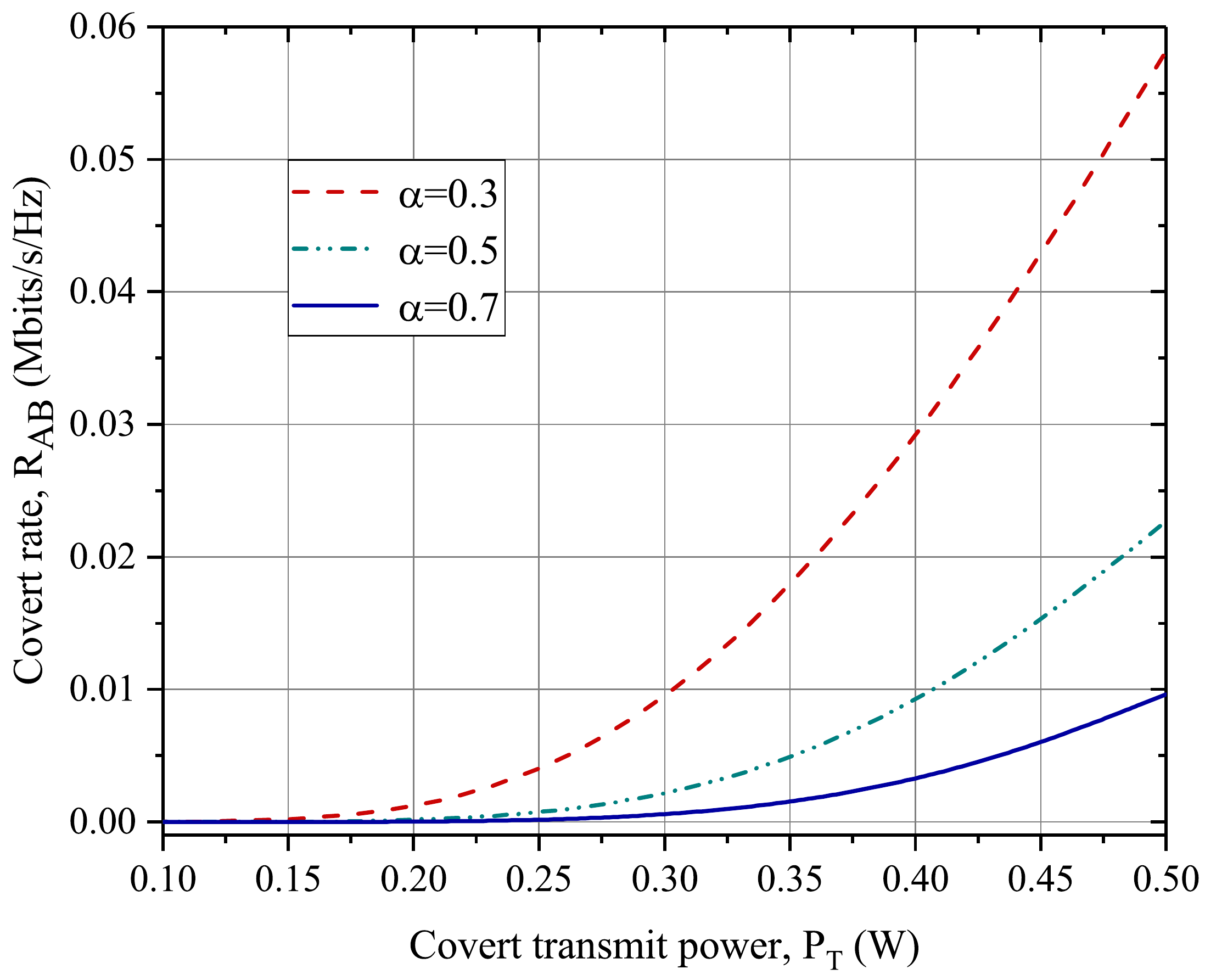}}
\subfigure[$R'_{AB}$ vs. $P_{T}$]{
\label{Fig.sub.2}
\includegraphics[width=0.4\textwidth]{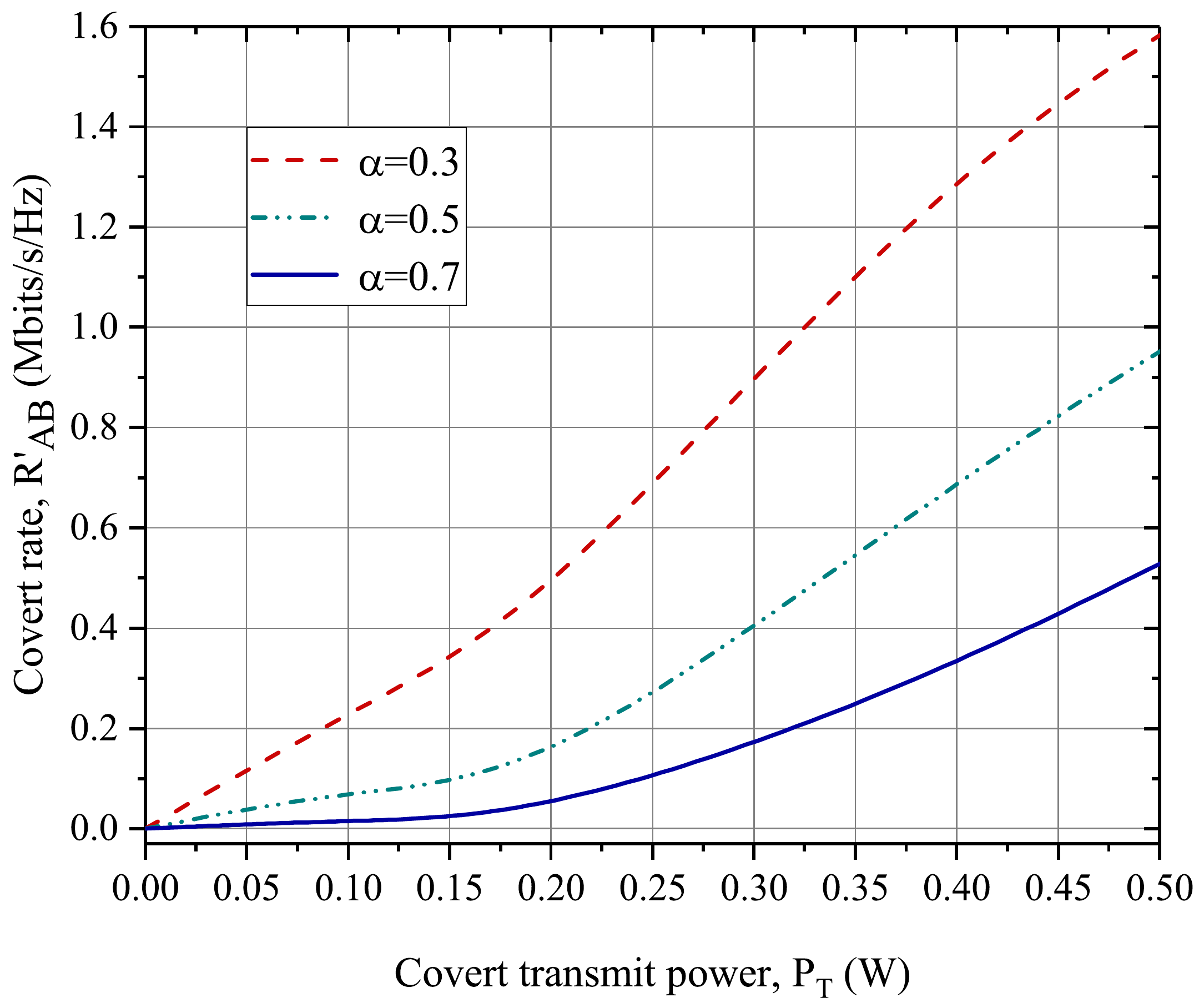}}
\caption{The impact of $P_{T}$ on covert rate.}
\label{Fig.lable}
\end{figure}
\begin{figure}[]
\centering
\subfigure[$R_{AB}$ vs. $P_{J}$]{
\label{Fig.sub.1}
\includegraphics[width=0.4\textwidth]{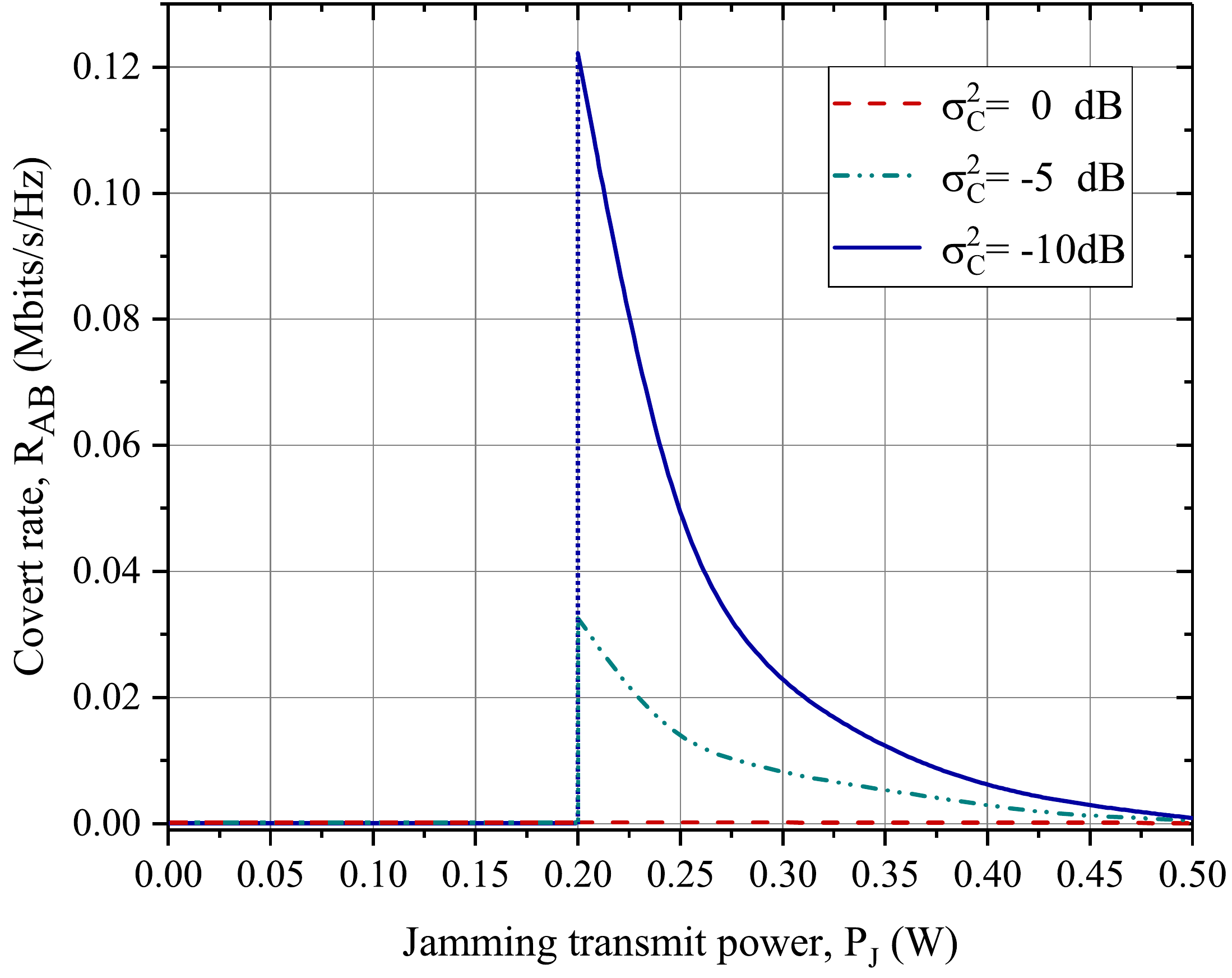}}
\subfigure[$R'_{AB}$ vs. $P_{J}$]{
\label{Fig.sub.2}
\includegraphics[width=0.4\textwidth]{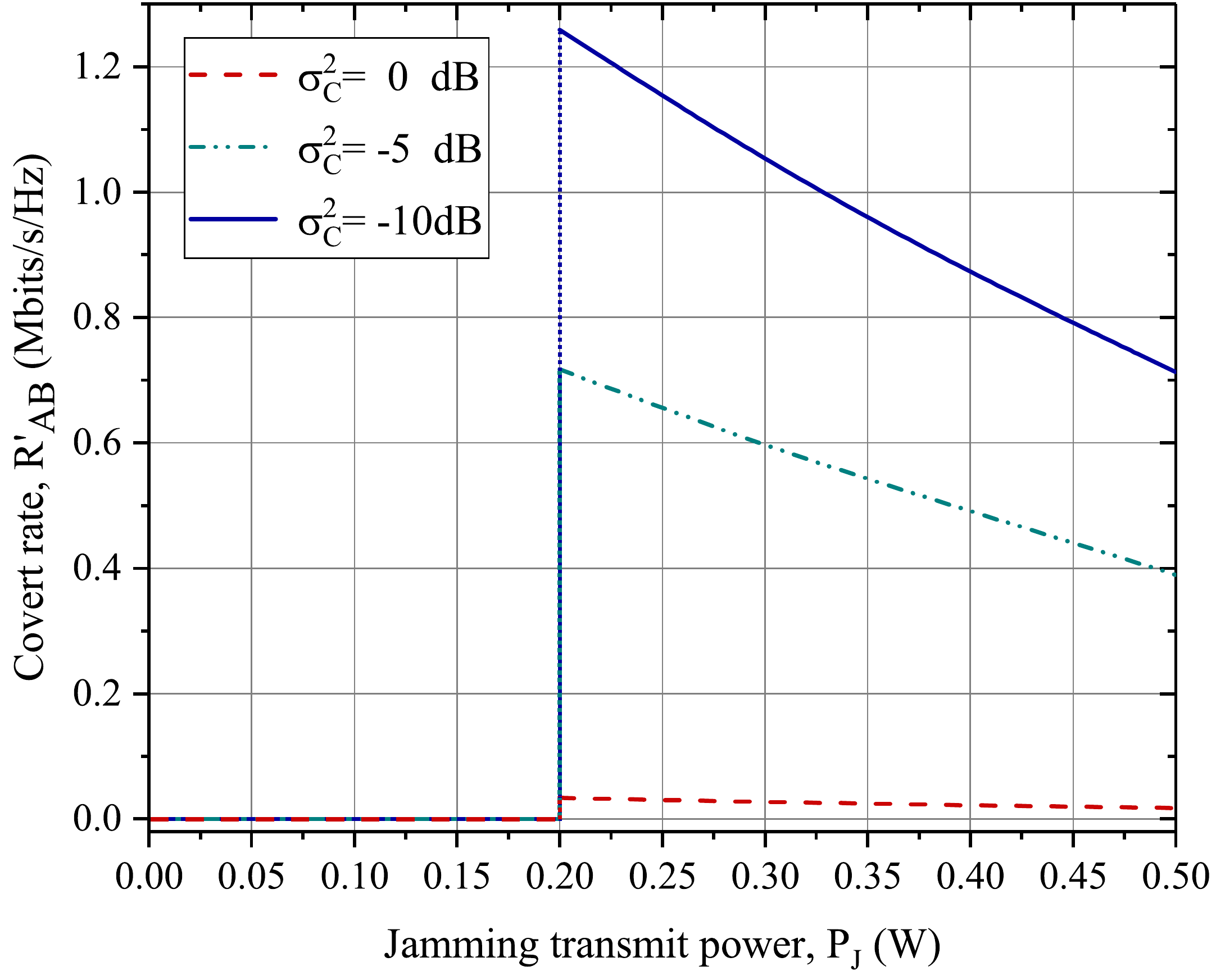}}
\caption{The impact of $P_{J}$ on covert rate.}
\label{Fig.lable}
\end{figure}

To investigate the impact of jamming transmit power $P_J$ on the covert rates under the two schemes,
we summarize in Fig. 5 how the covert rates vary with $P_J$ with the setting of $n=20$, $\alpha = 0.3$ and $\sigma_{C}^{2}=\sigma_{B}^{2}=\{0, -5, -10\}$ dB.

We can see from Fig. 5 that as $P_J$ increases,
the covert rate first increases and then decreases under each scheme.
This is because increasing $P_J$ has a two-fold effect on the covert rate.
It can confuse the detection of Willie, which leads to an increase in the covert rate.
Meanwhile, it can also interfere with the source-relay-destination links, which leads to a decrease in the covert rate.
As $P_J$ is relatively small, the former dominates the covert rate, and thus the covert rate increases with the increase of $P_J$.
As $P_J$ becomes larger, the latter dominates the covert rate, and thus the covert rate decreases as $P_J$ further increases.
Therefore, we can set a proper $P_J$ to improve the covert rate performance under each scheme.

\color{black}

\subsection{Performance Optimization and Comparison}

By optimizing the covert transmit power $P_T$, we explore the impact of covertness requirement $\varepsilon_{c}$ on the maximum covert rates under the two relay selection schemes as shown in Fig. 6, where $R^{*}_{AB}$ and $R^{'*}_{AB}$ represent the maximum covert rates under the
 RRS and MMRS, respectively.
For the setting of $n=20$, $\alpha = \{0.3, 0.5, 0.7\}$, and $\sigma_{C}^{2}=\sigma_{B}^{2}=-5$ dB,
we can observe from Fig. 6 that as $\varepsilon_{c}$ increases, the maximum covert rates increase under both the two schemes.
This is because the increase of $\varepsilon_{c}$ is equivalent to the increase of the probability with which the two-hop transmissions are detected by Willie.
This means that $P_T$ can increase, which leads to the increase of the maximum covert rates.

Finally, we conduct the performance comparison between the two relay selection schemes with jamming signals and these with no jamming signal,
as shown in Fig. 7.
It can be seen from Fig. 7 that the maximum covert rate under each scheme with jamming signals is larger than
that with no jamming signal.
This can be explained as follows.
With jamming signals, the covert requirement constraint is easier to be satisfied than that with no jamming signal.
Thus, the covert transmit power with jamming signals is larger than that with no jamming signal,
which leads to a larger maximum covert rate with jamming signals than no jamming signal.
Fig. 7 also illustrates that for each fixed setting of $\varepsilon_{c}$,  the maximum covert rate under MMRS is
larger than that under RRS.

\begin{figure}[]
\centering
\subfigure[$R_{AB}^{*}$ vs. $\varepsilon_{c}$]{
\label{Fig.sub.1}
\includegraphics[width=0.4\textwidth]{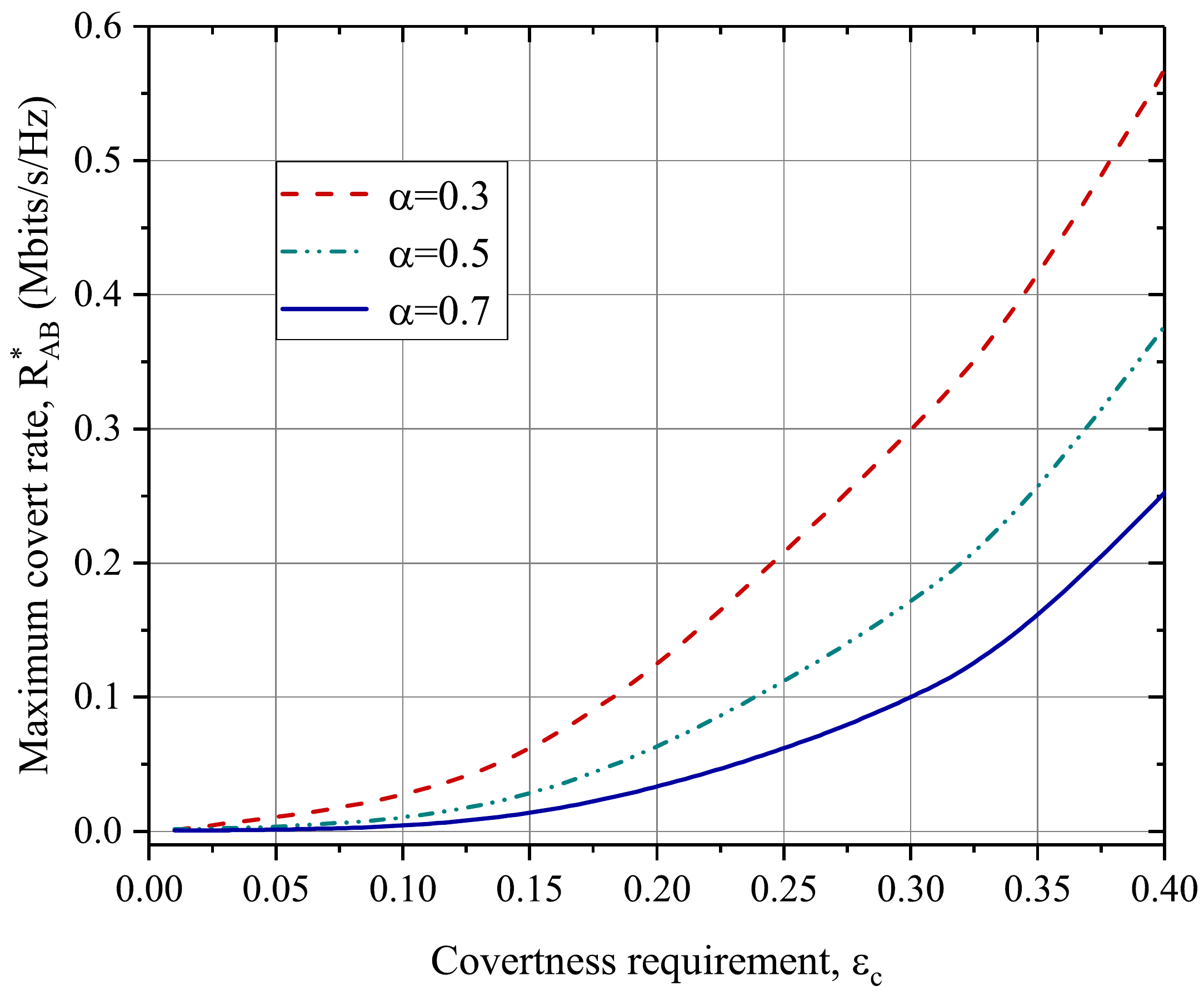}}
\subfigure[$R_{AB}^{'*}$ vs. $\varepsilon_{c}$]{
\label{Fig.sub.2}
\includegraphics[width=0.4\textwidth]{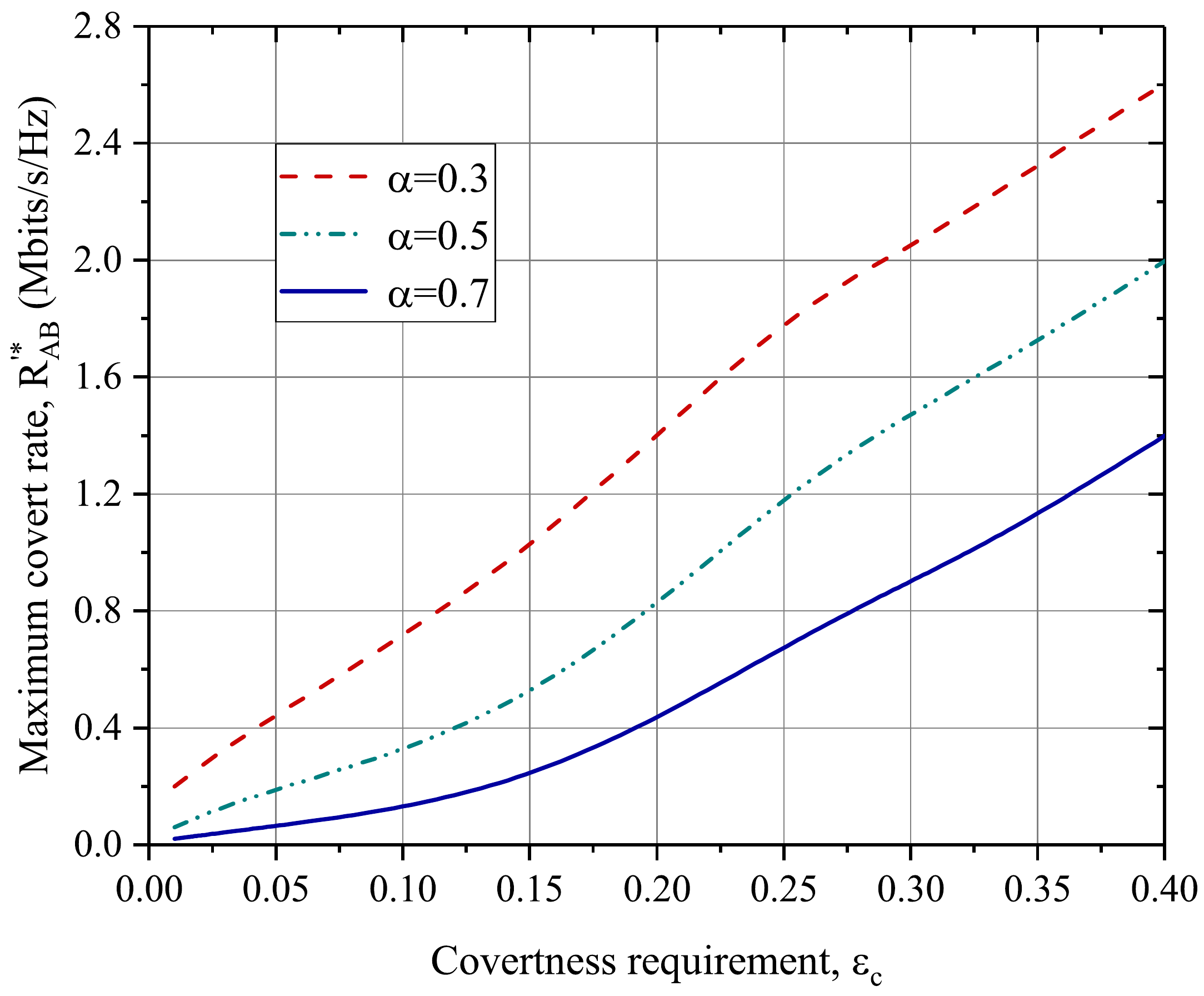}}
\caption{The impact of $\varepsilon_{c}$ on maximum covert rate.}
\label{Fig.lable}
\end{figure}

\begin{figure}[t]
\centering
\includegraphics[width=0.4\textwidth]{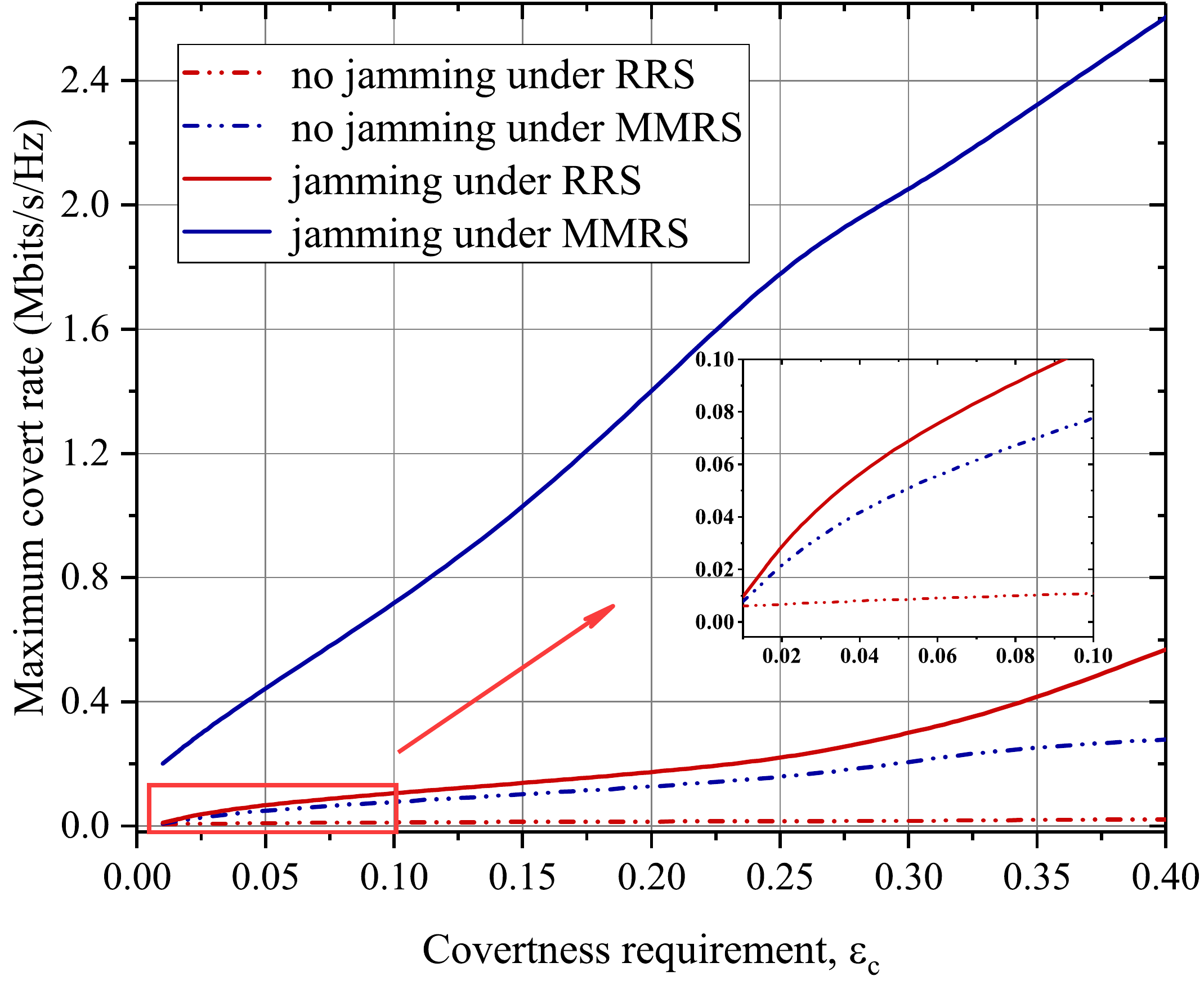}\\
\DeclareGraphicsExtensions.
\caption{Maximum covert rate comparison.}
\label{Fig.com}
\end{figure}

\section{Conclusion}

This paper explored the covert communications in a wireless relay system, where
two joint CJ and relay selection schemes are proposed for improving covert performance.
Based on each joint scheme, we developed a theoretical model to characterize the covert rate, and further maximize the covert rate by optimal transmit power control.
Finally, simulation/numerical results were provided to validate our theoretical models.
Specifically, increasing the covert transmit power can enhance the covert rate performance under each joint scheme.

%
%
%
%


%

\appendices




\ifCLASSOPTIONcaptionsoff
  \newpage
\fi

\end{document}